\documentclass[runningheads]{llncs}
\usepackage{graphicx}
\usepackage{amsmath}
\usepackage{amssymb}
\usepackage{mathrsfs}
\usepackage{subcaption}
\usepackage{lineno}

\newcommand{\vc}{{\mathbf{c}}}
\newcommand{\vd}{{\mathbf{d}}}

\newcommand{\vx}{{\mathbf{x}}}
\newcommand{\vy}{{\mathbf{y}}}
\newcommand{\vz}{{\mathbf{z}}}
\newcommand{\cPZ}{\mathcal{PZ}}
\newcommand{\cP}{\mathcal{Z}}
\DeclareMathOperator*{\argmin}{arg\,min}
\usepackage{hyperref}
\usepackage{graphicx}
\usepackage{listings}
\usepackage{epstopdf}
\usepackage{color}
 
\begin{document}
\title{On the Difficulty of Intersection Checking with Polynomial Zonotopes}
%
%
\author{Yushen Huang\inst{1} \and
Ertai Luo\inst{1} \and
Stanley Bak\inst{1} \and
Yifan Sun\inst{1}}
\authorrunning{Y. Huang et al.}
%
\institute{Stony Brook University, Stony Brook NY 11790, USA \\
\email{\{yushen.huang,ertai.luo,stanley.bak,yifan.sun\}@stonybrook.edu}}
\maketitle              
\begin{abstract}
Polynomial zonotopes, a non-convex set representation, have a wide range of applications from real-time motion planning and control in robotics, to reachability analysis of nonlinear systems and safety shielding in reinforcement learning.
Despite this widespread use, a frequently overlooked difficulty associated with polynomial zonotopes is intersection checking. 
Determining whether the reachable set, represented as a polynomial zonotope, intersects an unsafe set is not straightforward.
In fact, we show that this fundamental operation is NP-hard, even for a simple class of polynomial zonotopes.

\vspace{1em}
The standard method for intersection checking with polynomial zonotopes is a two-part algorithm that overapproximates a polynomial zonotope with a regular zonotope and then, if the overapproximation error is deemed too large, splits the set and recursively tries again.
Beyond the possible need for a large number of splits, we identify two sources of concern related to this algorithm: (1) overapproximating a polynomial zonotope with a zonotope has unbounded error, and (2) after splitting a polynomial zonotope, the overapproximation error can actually increase.
Taken together, this implies there may be a possibility that the algorithm does not always terminate.
We perform a rigorous analysis of the method and detail necessary conditions for the union of overapproximations to provably converge to the original polynomial zonotope.

\end{abstract}
\section{Introduction}
Set-based analysis is the foundation of many formal analysis approaches including abstract interpretation methods for software~\cite{cousot1977abstract} and reachability analysis methods for cyber-physical and hybrid systems~\cite{althoff2021set}.
The usefulness of a set representation is determined by what operations can be efficiently supported.

For safety verification, one fundamental operation is intersection checking; does the set of possible states intersect the set of unsafe states?
In this context, one common way to represent sets is using zonotopes~\cite{eppstein1995zonohedra,girard2005reachability}, which are affine transformations of a unit box. 
Zonotopes offer a compact representation, efficiently encode linear transformations, and support linear-time optimization.
However, zonotopes cannot represent non-convex sets and so are less useful when a nonlinear operation is applied to a set.
In contrast, polynomial zonotopes \cite{althoff2013reachability} are closed under polynomial maps and can therefore exactly represent more complex sets. 
Polynomial zonotopes can be considered as polynomial transformations of a unit box.
The two representations are illustrated in Figure~\ref{fig:Single PZ Filled}.

One drawback of polynomial zonotopes is that intersection checking is significantly more complex than with zonotopes.
Although it is known that the characterization of solutions of general nonlinear equations with box-constrained domains is NP-hard~\cite[Sec. 4.1]{jaulin2001interval}, is the problem easier for polynomial zonotopes, since the transformation is always a polynomial? Would the problem become easier if we only check for halfspace intersections or if we only consider simple polynomials?
In this work we prove that intersection checking is NP-hard for polynomial zonotopes, regardless of such attempts at simplification.

Setting aside the worst-case time complexity, the existing algorithm proposed to check for intersections, as well as perform plotting, is based on a combination of overapproximation using zonotopes and refinement using splitting~\cite{kochdumper2022extensions,bak2022reachability}.
Is this algorithm guaranteed to converge to the true set, even given infinite runtime?
We identify two sources of concern: (i) the overapproximation of a polynomial zonotope with a zonotope can have unbounded error, and (ii) the error of polynomial zonotope overapproximation can actually \emph{increase} after splitting is performed.
This work analyses the proposed algorithm in detail, and derives fairness conditions that are sufficient to prove the algorithm provably converges. 

\begin{figure}[t]
   \centering
   \includegraphics[width=0.75\textwidth]{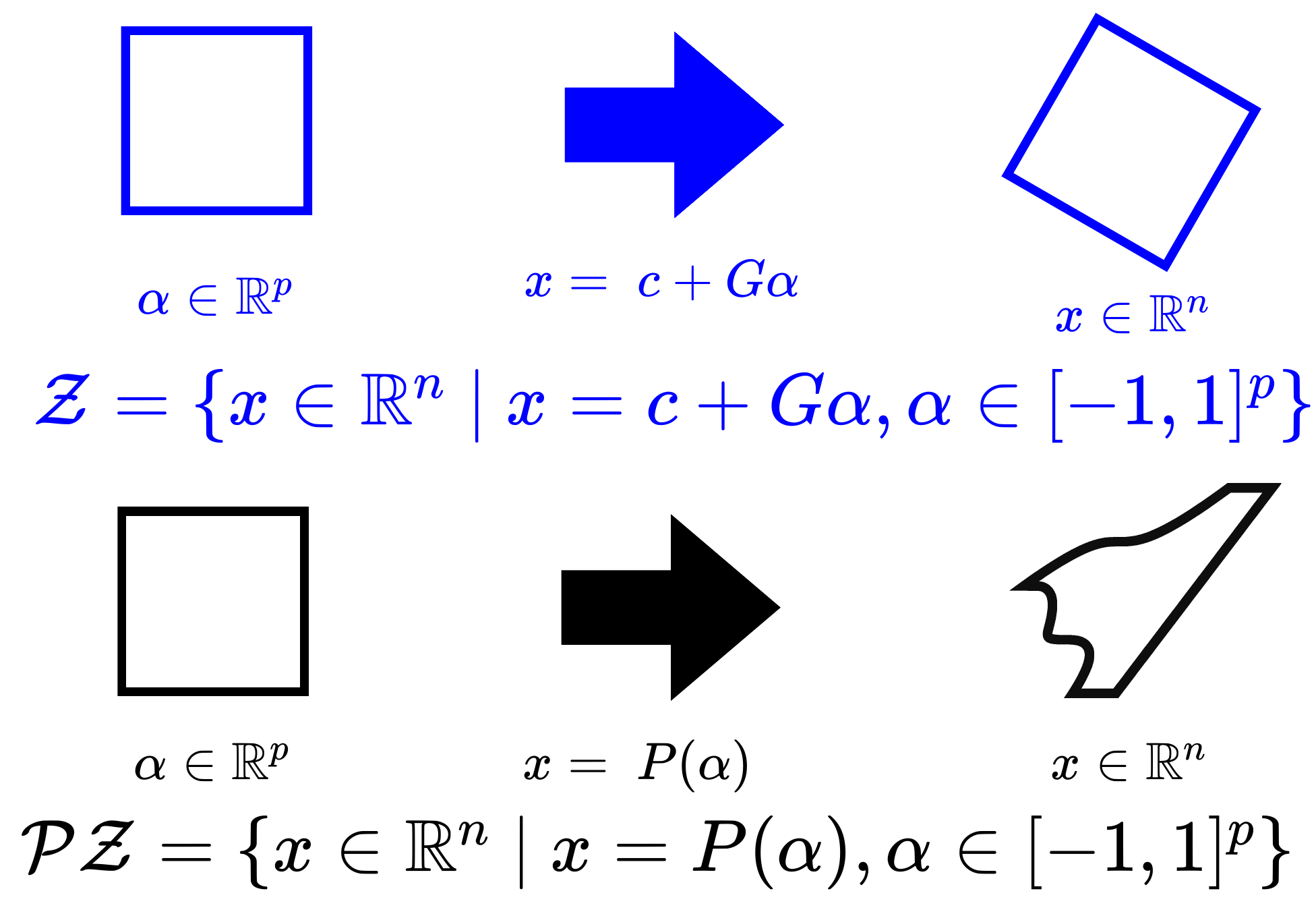}
   \caption{A zonotope (blue, top) is a convex $n$-dimensional set represented as an affine transformation of a unit box in $p$ dimensions. A polynomial zonotope (black, bottom) is a possibly non-convex $n$-dimensional set represented as a polynomial transformation of a unit box in $p$ dimensions ($P(\cdot)$ is a polynomial). 
   }.
   \label{fig:Single PZ Filled}
\end{figure}

\vspace{1em}
\noindent
\textbf{Practical Example.} While the contributions of this work are theoretical in nature, they are grounded in practical issues the authors observed while working with polynomial zonotopes.
Figure~\ref{fig:split_comparison} shows a plot of a 2-d projection of a polynomial zonotope, obtained when computing the reachable set of an uncertain time-varying system~\cite{Luo2023} using the overapproximate and split algorithm from the CORA tool~\cite{althoff2015introduction}.
Splitting seems to have diminishing returns, as the light gray overapproximation of the polynomial zonotope remains far from the true boundary (red points), even when the algorithm runs for over an hour.

\vspace{1em}
\noindent
\textbf{Contributions.} The key contributions of this paper are:
\begin{itemize}
    \item[$\bullet$] We prove that polynomial zonotope intersection checking
    is NP-hard, even for simple halfspace constraints and bilinear polynomials (Section~\ref{sec:nphard}).
    \item[$\bullet$] We review the standard intersection-checking algorithm, and demonstrate two sources of concern, that 
        overapproximation error is unbounded and that overapproximation error can increase after splitting (Section~\ref{sec:alg_concerns}). 
    \item[$\bullet$] We provide conditions where the polynomial zonotope refinement algorithm provably converges to the original polynomial zonotope (Section~\ref{sec:convergence}).
\end{itemize}

First, we review preliminaries and formally define zonotopes and polynomial zonotopes in Section~\ref{sec:prelim}.

\begin{figure}[t]
    \centering
    \begin{subfigure}{0.49\textwidth}
        \centering
        \includegraphics[width=\linewidth]{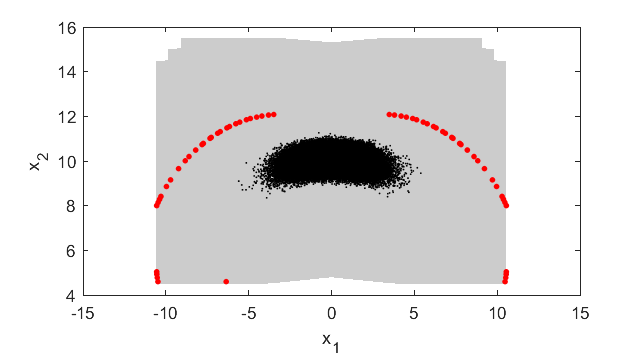}
        \caption{2 splits (6 sets), 0.02 seconds}
    \end{subfigure}%
    \hfill
    \begin{subfigure}{0.49\textwidth}
        \centering
        \includegraphics[width=\linewidth]{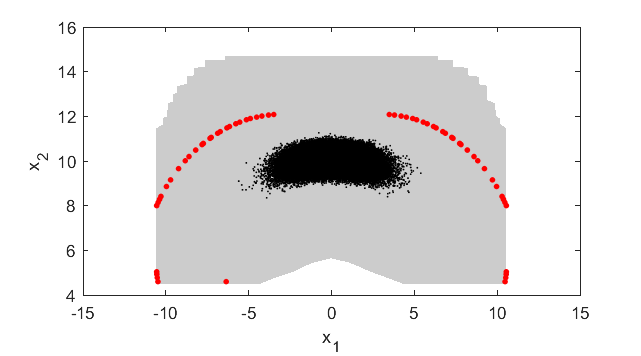}
        \caption{10 splits (476 sets), 1.2 seconds}
    \end{subfigure}
    \vskip\baselineskip
    \begin{subfigure}{0.49\textwidth}
        \centering
        \includegraphics[width=\linewidth]{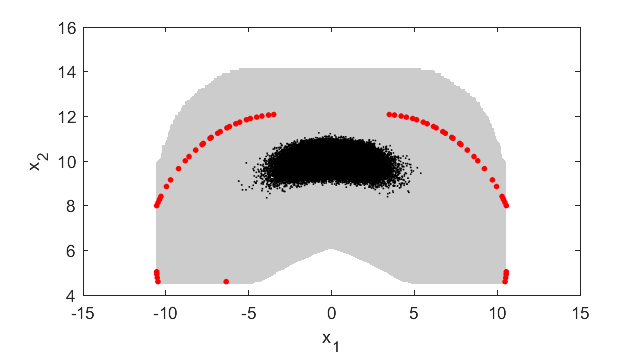}
        \caption{20 splits (14K sets), 51 seconds}
    \end{subfigure}%
    \hfill
    \begin{subfigure}{0.49\textwidth}
        \centering
        \includegraphics[width=\linewidth]{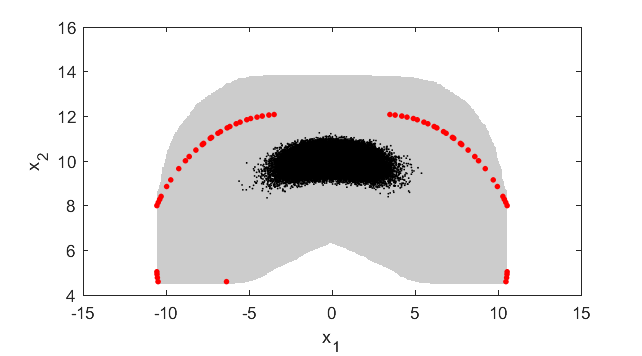}
        \caption{40 splits (310K sets), 1.2 hours}
    \end{subfigure}
    \caption{
    Plotting a polynomial zonotope using the overapproximate and split algorithm (light gray) does not converge to the true set even after an hour of computation time.
    The red dots are the true boundary points and the black dots are random samples.}
    \label{fig:split_comparison}
\end{figure}

\section{Preliminaries}
\label{sec:prelim}

\textbf{Notation.} The set $\mathbb{R}^n$ is an $n$-dimensional real space and $\mathbb{Z}_{\geq 0}$ is the set of all non-negative integers. 
Given a matrix $A \in \mathbb{R}^{n \times m}$, let $A(i,\cdot)$ be the $i$-th row of the matrix and $A(\cdot,j)$ be the $j$-th column of the matrix.
Given vector $\vx \in \mathbb{R}^n$, the $i$-th component of the vector is referred to as $x_i$ and the (one-) norm of the vector is $\Vert \vx \Vert = \sum_{i=1}^{n} \vert x_i \vert $. 
Given a set $S \subseteq \{1,2,\cdots, m \}$, we denote $A(\cdot, S)$ as the matrix consisting of the row index belonging to $S$. For example if $S = \{1,3\}$, then $A(\cdot,S) = \begin{bmatrix}
    A(\cdot,1) & A(\cdot,3)
\end{bmatrix}$. Similarly, given a set $S \subseteq \{1,2,\cdots, n \}$, we denote $A(S,\cdot)$ as the matrix that consists of the column index belong to $S$. 
We call the $n \times n$ identity matrix $I_n$. 
Given two sets $A$ and $B$, the Minkowsiki sum of is written as $A \oplus B = \big \{ \vz ~\big |~ \vz = \vx + \vy, \ \vx \in A, \ \vy \in B \big \}$.

\vspace{1em}
\noindent
We start by defining zonotopes and polynomial zonotopes more formally.


\begin{definition}[Zonotope] 
Given a center $\vc \in \mathbb{R}^n$ and generator matrix $G \in \mathbb{R}^{n \times p} $, a zonotope is the set
\begin{align*}
    \cP = \Bigg \{ \vc + \sum_{j=1}^{p} \alpha _jG(\cdot, i) ~ \bigg|~  \alpha_j \in [-1,1] \Bigg \}.
\end{align*}
\end{definition}
We refer to a zonotope using the shorthand notation $\cP = \langle\vc, G\rangle_{\cP}$. 
Note in the illustration in Figure~\ref{fig:Single PZ Filled}, $\alpha$ was a $p$-dimensional point, whereas in the definition we refer to each element as a scalar $\alpha_j$, which we call a \emph{factor}.

As mentioned in the introduction, a polynomial zonotope is a polynomial transformation of a $p$-dimensional unit hypercube.
Following the sparse formulation of polynomial zonotopes~\cite{kochdumper2020sparse}, we explicitly split the factors from the $p$-dimensional point into two sets $\alpha \in \mathbb{R}^r$ and $\beta \in \mathbb{R}^q$, with $p = r + q$. 
The $\beta$ factors are called \emph{independent} and only occur in terms by themselves and with an exponent of one, whereas the $\alpha$ factors are called \emph{dependent} and are allowed to multiply other (dependent) factors within the same term or have higher powers.



\begin{definition}[Polynomial Zonotope]\label{polynomial_zonotope_def}
Given center $c \in \mathbb{R}^n$, dependent factor generator matrix $G_D \in \mathbb{R}^{n \times h}$, independent factor generator matrix $G_I \in \mathbb{R}^{n \times q}$, and exponent matrix $E \in \mathbb{Z}_{\geq 0}^{r \times h}$, a polynomial zonotope is the set:
\begin{equation*}
  	\begin{split}
    \mathcal{PZ} \hspace{-2pt}= \hspace{-2pt}\bigg\{ & c+\hspace{-2pt}\sum _{i=1}^h \bigg( \prod _{k=1}^r \alpha _k ^{E(k,i)} \bigg) G_{D(\cdot,i)}  + \sum _{j=1}^{q} \beta _j\,G_{I(\cdot,j)}~ \bigg|~\alpha_k, \beta_j \in [-1,1] \bigg\}.
    \end{split}
\end{equation*}
\end{definition}

The intuition why we separate the dependent factors from the independent factors is that in this way the polynomial zonotope can always be written as a Minkowski sum of two sets:
 $$
 \cPZ = \cP_{I} \oplus \cPZ_{D}
 $$
where
\begin{align*}
\cP_{I} =& \bigg\{ c + \sum _{j=1}^{q} \beta _j\,G_{I(\cdot,j)} ~\bigg|~ \beta_j \in [-1,1] \bigg\}, \\
%
%
\cPZ_{D} =& \Bigg \{ \sum_{i=1}^{h} \bigg( \prod_{k=1}^{r} \alpha_k^{E(k,i)} \bigg) G_{D(\cdot,i)} ~\bigg|~ \alpha_k \in [-1,1]  \Bigg \}.
\end{align*}
This splits the general polynomial zonotope into the independent part $\cP_{I}$ which is a zonotope, and the polynomial zonotope $\cPZ_D$ which contains only terms where factors multiply each other or have higher powers.
For intersection checking, the complexity arises from the dependent part and so we will often use a form with only dependent terms like $\cPZ_D$.
In this paper, such a polynomial zonotope, with only dependent terms, will be written using the shorthand notation~$\langle G_D, E \rangle_{\cPZ}$.

\begin{example} (Figure~\ref{fig:Single PZ Filled}, bottom)\label{PZ_example}
    Consider a polynomial zonotope defined as:
    $$
        \cPZ = \Bigg \{  \begin{bmatrix}
           4 \\
           4
         \end{bmatrix} + \beta_1 \begin{bmatrix}
           1 \\
           0
         \end{bmatrix} + \alpha_1 \begin{bmatrix}
           2 \\
           0
         \end{bmatrix} + \alpha_2 \begin{bmatrix}
           1 \\
           2
         \end{bmatrix} + \alpha_1^3\alpha_2 \begin{bmatrix}
           2 \\
           2
         \end{bmatrix} ~\bigg|~ \alpha_i, \beta_i \in [-1,1]  \Bigg \}.
    $$
    In this example, $q = 1, r = 2$ and $h = 3$. 
    We can split $\cPZ$ into the Minkowski sum of two sets $\cPZ = \cP_{I} \oplus \cPZ_{D}$, with
     \begin{align*}
         \cP_{I} = \bigg \langle \begin{bmatrix}
             4 \\ 4
         \end{bmatrix}, \begin{bmatrix}
             1 \\ 0
         \end{bmatrix} \bigg \rangle_{\cP} \quad \cPZ_{D} =  \bigg \langle \begin{bmatrix}
            2& 1& 2\\
          0 & 2 & 2
         \end{bmatrix}, \begin{bmatrix}
           1 & 0 &  3\\
            0 & 1 & 1
         \end{bmatrix} \bigg \rangle_{\cPZ}.
     \end{align*}

\end{example}
\section{Intersection Checking is NP-Hard}
\label{sec:nphard}

Given a polynomial zonotope $\cPZ = \cP_{I} \oplus \cPZ_{D}$ and a linear objective direction $\vd$ the \emph{polynomial zonotope optimization problem} computes the value:
\begin{align}\label{Polynomial_Optimization}
\min_{\vx \in \cPZ} \vx^T \vd = 
\min_{\vx_1 \in \cP_{I}} \vx_1^T \vd + 
\min_{\vx_2 \in \cPZ_{D}} \vx_2^T \vd.
\end{align}
%
If we want to check whether a polynomial zonotope $\cPZ$ has intersection with a halfspace $\mathcal{H} = \{ \vx ~|~ \vx^T \vd \leq c \}$, we only need to check if the solution of \eqref{Polynomial_Optimization} is larger than $c$. 
Since linear optimization of zonotopes is efficient, the main challenge lies in computing the optimal value in the polynomial zonotope of dependent terms $\cPZ_{D}$. 
This is illustrated in the following example:
\begin{example}
    Consider checking if the polynomial zonotope $\cPZ$ from Example \ref{PZ_example} has an intersection with the halfspace 
    \[
    \mathcal{H} = \Big \{ \vx \in \mathbb{R}^2 ~~ \Big| ~~  \vx^T   \begin{bmatrix}
           1 \\ 
           1 
         \end{bmatrix}\leq 0 \Big \}.
         \]
         
         In order to check this, we only need to check whether the solution of the problem below is larger or equal to 0 : 
         \begin{align*}
          \min_{\vx \in \cPZ}  \ \vx^T \begin{bmatrix}
           1 \\
           1 
         \end{bmatrix}   &= \min_{\alpha_k, \beta_k \in [-1,1] }  \left (\begin{bmatrix}
           4 \\
           4
         \end{bmatrix} + \beta_1 \begin{bmatrix}
           1 \\
           0
         \end{bmatrix} + \alpha_1 \begin{bmatrix}
           2 \\
           0
         \end{bmatrix} + \alpha_2 \begin{bmatrix}
           1 \\
           2
         \end{bmatrix} + \alpha_1^3\alpha_2 \begin{bmatrix}
           2 \\
           2
         \end{bmatrix} \right )^T \begin{bmatrix}
           1 \\
           1 
         \end{bmatrix} \\
         &= \min_{\alpha_k, \beta_k \in [-1,1]} 8 + \beta_1 + 2\alpha_1 + 3\alpha_2 + 4\alpha_1^3\alpha_2  \\
         &= \min_{\beta_k \in [-1,1]} 8 + \beta_1 + \min_{\alpha_k \in [-1,1]} 2\alpha_1 + 3\alpha_2 + 4\alpha_1^3\alpha_2  \\
         &= 7 + (-5) = 2
         \end{align*}
         The minimum value is larger than $0$ so there is no intersection with the halfspace. 
\end{example}

As we saw in the above example, solving the optimization problem in \eqref{Polynomial_Optimization} can be used to check whether a polynomial zonotope has an intersection with a halfspace.
Furthermore, the generators of the polynomial zonotope can be projected onto the optimization direction resulting in a 1-D optimization problem.
How difficult is this problem?
As mentioned in the introduction, the full characterization of solutions of nonlinear equations given box domains is NP-hard~\cite[Sec. 4.1]{jaulin2001interval}.
In fact, even if we restrict ourselves to optimization and only consider \emph{bilinear polynomial zonotopes}---the simplest class of polynomial zonotopes with two variables per term each with an exponent of one---the problem is still NP-hard, which we show next.

First we introduce the 1-D \emph{multi-affine polynomial optimization problem}.
\begin{definition}\label{multi-affine-def}
    Consider the polynomial 
    defined as:
    \begin{align}\label{MultiAffine}
     p(x_1,x_2,\cdots,x_n) = \sum_{I \subseteq \{1,2,\cdots,n \}} a_I \prod_{i \in I} x_i   
    \end{align}
    The 1-D multi-affine polynomial optimization problem is:
    \begin{align*}
        \min_{\substack{x_1,...,x_n \\ x_i \in [-1,1] }} p(x_1,x_2,\cdots,x_n).
    \end{align*}
\end{definition}

Since all variables in a multi-affine optimization problem have an exponent of one, the partial derivative along each variable cannot change sign.
This means that the optimal value must occur on one of the corners of the $n$-dimensional box of the domain and it is sufficient to consider this finite set when optimizing.
    \begin{align*}
        \min_{\substack{x_1,...,x_n \\ x_i \in [-1,1] }} p(x_1,x_2,\cdots,x_n) = \min_{\substack{x_1,...,x_n \\x_i \in \{-1\} \cup \{1\}}} p(x_1,x_2,\cdots,x_n)
    \end{align*}
Note that the polynomial zonotope motivating our work in Figure~\ref{fig:split_comparison} was a multi-affine polynomial zonotope; we obtained the true boundary points shown in red using a version of this corner enumeration strategy.

The simplest type of non-trivial multi-affine optimization problem has two variables per term, since any terms with a single variable could be optimized by simply looking at the sign of $a_I$ similar to optimization methods for zonotopes.
We call this a \emph{bilinear optimization problem}, which corresponds to optimization of a linear objective function over the dependent factors part of a bilinear polynomial zonotope, which has the corresponding restrictions on its terms
\begin{align}
\label{eq:bilinear_opt}
    \min_{    \substack{x_1,...,x_n \\ x_i \in \{-1\} \cup \{1\} }} \sum_{i=1}^{n} \sum_{j=i+1}^{n} a_{i,j} x_i x_j.
\end{align}

\begin{theorem}
Optimization over bilinear polynomial zonotopes is NP-complete.
\end{theorem}
\begin{proof}
    We show that if we could solve the bilinear optimzation problem from~\eqref{eq:bilinear_opt}, then we could also solve the \emph{minimum edge-deletion graph bipartization} problem, which is NP-complete~\cite{yannakakis1978node,garey1974some}.
    The minimum edge-deletion graph bipartization problem is the problem of computing the minimum number of edges that must be deleted so that an undirected graph $G$ becomes a bipartite graph\footnote{In a bipartite graph, there are two groups of vertices, and edges are only allowed between the groups, not within each group.}.
    Let $G = (V,E)$ be an arbitrary undirected graph with vertices $ V = \{1,2,\ldots, n\}$ and edges $E$, where an edge $e \in E$ connecting vertices $i$ and $j$ is represented as $e = (i,j)$, with convention $i < j$.
    Let  $\delta_G$ be the least number of edges we need to remove to make graph $G$ bipartite. 
    In \eqref{eq:bilinear_opt}, we assign $a_{i,j} = \left \{\begin{array}{cc}
       \frac{1}{2}  & (i,j)\in E \\
       0  & (i,j) \notin E
    \end{array} \right .$. 
    
    Now define the assignment of the variables corresponding to the optimal solution of~\eqref{eq:bilinear_opt} as
    $$
    x_1^{*},\cdots,x_n^{*} = \argmin_{x_i \in \{-1\} \cup \{1\}} \sum_{i=1}^{n} \sum_{j=i+1}^{n} a_{i,j} x_i x_j 
    $$
    and define the value $\delta$ as 
    \begin{align} \label{eq:opt_delta}
     \delta = \frac{\vert E \vert}{2} + \sum_{i=1}^{n} \sum_{j=i+1}^{n} a_{i,j} x_i^* x_j^* 
    \end{align}
    Now consider a biartite partitioning $V = V_1\cup V_2$, where $V_1 = \big \{ i \in V ~\big |~ x_i^* = -1 \big \}$ and $V_2 = \big \{ i \in V ~\big |~ x_i^* = 1 \big \} $.
    Let $\tilde{E}$ be the set of edges $ = (i,j)$ where either  $i,j \in V_1$ or $i,j \in V_2$; these are the edges to be removed such that $G$ becomes a bipartite graph. 
    By the definition of $\delta_G$, we must have  $\vert \tilde{E} \vert \geq \delta_G$. Now since $x_i^*x_j^* = 1$ when $(i,j) \in \tilde{E}$, and  $x_i^*x_j^* = -1$ when $(i,j) \in  E/ \tilde{E}$ we have
    \begin{align*}
           \vert \tilde{E} \vert  & = \sum_{(i,j) \in \tilde{E}} x_i^*x_j^*   \\
           &=  \underbrace{\frac{1}{2}\left ( \sum_{(i,j) \in \tilde{E}}x_i^*x_j^* - \sum_{(i,j) \in E/ \tilde{E}}x_i^*x_j^* \right)}_{\frac{\vert E \vert }{2}} +\underbrace{ \frac{1}{2}\left ( \sum_{(i,j) \in \tilde{E}}x_i^*x_j^* +  
            \sum_{(i,j) \in E / \tilde{E}}x_i^*x_j^*\right )}_{\sum_{i=1}^{n}\sum_{j = i +1}^{n}a_{i,j}x_i^*x_j^*}  \\
           &= \frac{\vert E \vert}{2} + \sum_{i=1}^{n}\sum_{j = i +1}^{n}a_{i,j}x_i^*x_j^* \\ 
           & = \delta 
    \end{align*}
    Hence $\delta \geq \delta_G$.
    
    Next, define $E_G$ to be the smallest set of edges that need to removed to make $G$ bipartite, so that $\vert E_G \vert = \delta_G$.
    The graph $ \tilde{G} = (V,E/E_G)$ is bipartite, so we can partition the graph $\tilde{G}$ into the two bipartite sets $V_1^*$ and $V_2^*$ such that there are only edges are between $V_1^*$ and $V_2^*$. Now define:
    $$
    x'_i = \left \{ \begin{array}{cc}
       -1,  & i \in V_1^*, \\
        1, &  i \in V_2^*.
    \end{array}
    \right.$$
    Then, since $\delta$ comes from the solution of the minimization problem in~\eqref{eq:opt_delta}:
    \begin{align*}
        \delta \leq \frac{\vert E \vert}{2} + \sum_{i=1}^{n} \sum_{j=i+1}^{n} a_{i,j} x'_i x'_j &= \frac{\vert E \vert}{2} + \frac{1}{2}
        \underbrace{\sum_{(i,j) \in E/E_G} x'_ix'_j}_{(\vert E\vert  -\vert E_G\vert)(-1)} 
        + 
        \frac{1}{2} \underbrace{ \sum_{(i,j) \in E_G} x'_i x'_j}_{\vert E_G \vert} 
        \\
        &= \frac{\vert E \vert}{2} - \left (\frac{\vert E\vert  -\vert E_G\vert } {2}\right )  + \frac{\vert E_G \vert }{2} \\
        &= \vert E_G \vert = \delta_G
    \end{align*}
    Therefore $\delta \leq \delta_G$ and combining both parts $\delta = \delta_G$. Hence finding the solution of \emph{minimum edge-deletion graph bipartization} problem reduces to solving  \eqref{eq:bilinear_opt}. \qed
\end{proof}

\begin{corollary}\label{corollary:nphard}
Polynomial zonotope intersection checking is NP-hard.
\end{corollary}
\begin{proof}
    Since optimization of bilinear polynomial zonotopes is NP-complete, intersection checking of bilinear polynomial zonotopes is also NP-complete.
    As these are a type of polynomial zonotope, halfspace intersection checking of general polynomial zonotopes is also at least as difficult, and so is NP-hard.
    %
    \qed
\end{proof}

\section{The Overapproximate and Split Algorithm}
\label{sec:alg_concerns}
While Corollary~\ref{corollary:nphard} showed that checking the intersection between a polynomial zonotope and another set can be a difficult problem, what algorithm is used in practice?
The existing method~\cite{kochdumper2022extensions,bak2022reachability} consists of two steps.
In step one, the polynomial zonotope is overapproximated using a zonotope.
If this zonotope overapproximation does not intersect the other set, then the smaller polynomial zonotope does not intersect the other set either and the algorithm terminates.
Otherwise, a point is sampled from inside the polynomial zonotope and tested if it is inside the other set\footnote{The specific sample point is not important for convergence, although it is typically heuristically derived from the zonotope overapproximation.}.
If so, a witness point for the intersection has been found and the algorithm terminates.
If neither of these are applicable, in step two, the algorithm divides the polynomial zonotope into two smaller polynomial zonotopes and repeats from step one recursively.

The algorithm used to plot a polynomial zonotope is similar, using a recursive depth bound and then plotting the zonotope overapproximations at the tree leaves.
For this method to obtain high precision, we may need to split the polynomial zonotope into a large number of smaller pieces, compute the zonotope approximation for each piece, and then take the union of those zonotopes to serve as the overapproximation of the original polynomial zonotope. 
Figure \ref{fig:split_example} shows a visualization of the overapproximate and split intersection algorithm.
\begin{figure}[t]
    \centering
    \begin{subfigure}{0.49\textwidth}
        \centering
        \includegraphics[width=\linewidth]{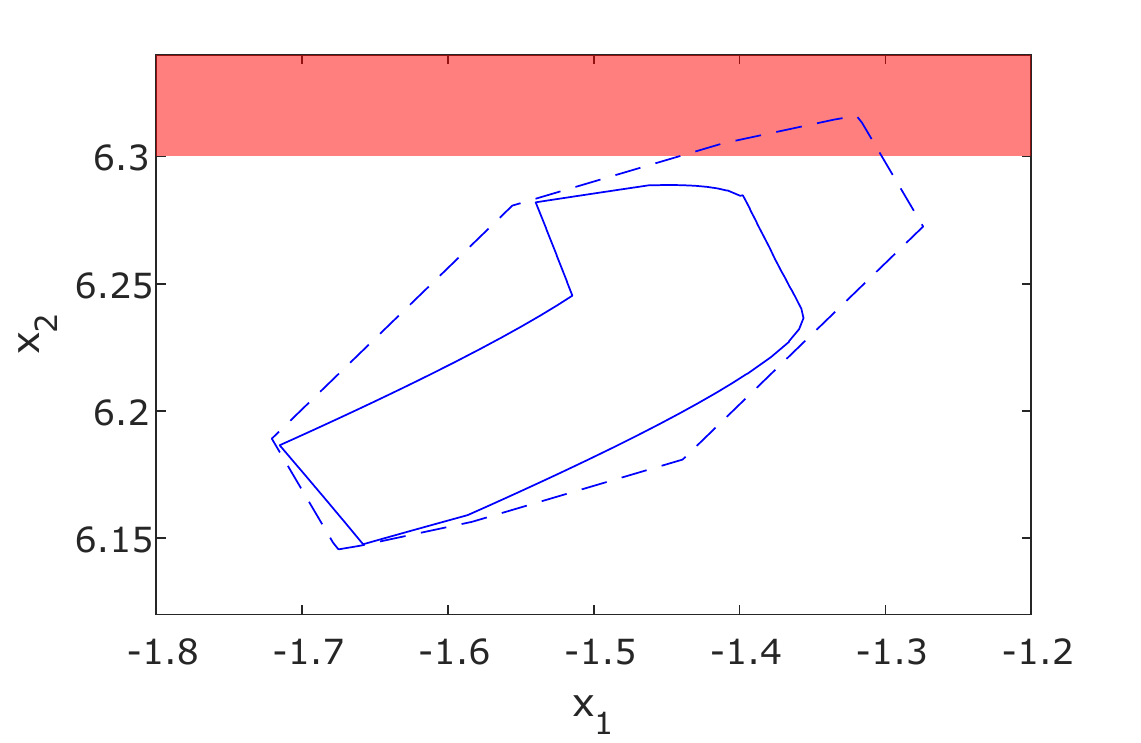}
    \end{subfigure}%
    \hfill
    \begin{subfigure}{0.49\textwidth}
        \centering
        \includegraphics[width=\linewidth]{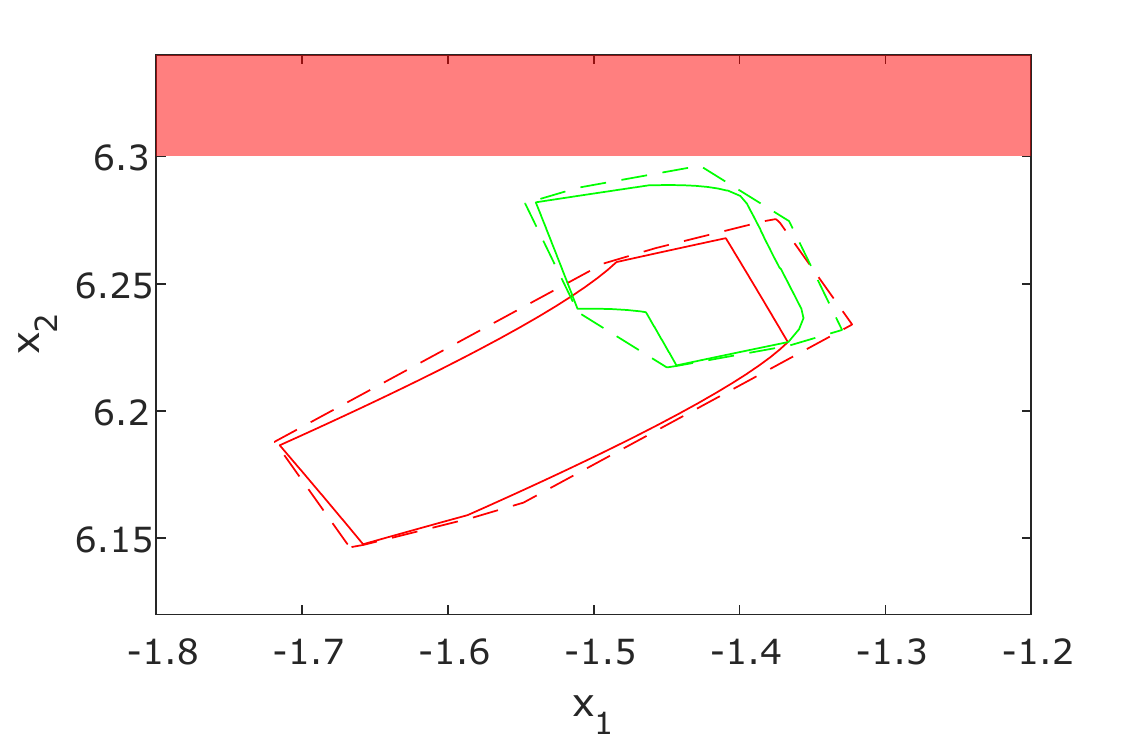}
    \end{subfigure}
    \caption{An illustration of the overapproximate and split intersection algorithm, where the zonotope over-approximation (dashed line) of the original polynomial zonotope is too conservative (left) while after splitting the zonotope over-approximation of the two split polynomial zonotopes is accurate enough to show there is no intersection with the red shaded region (image from~\cite{bak2022reachability}).}
    \label{fig:split_example}
\end{figure}   

\subsection{Algorithm Definition}
The algorithm consists of two steps: (i) overapproximate and (ii) split.

\vspace{1em}
\noindent
\textbf{Overapproximation Step.} The key observation motivating the overapproximation step is that since each factor $\alpha_k \in [-1,1]$, the product of factors in each term in outer sum is also in $[-1,1]$:
$$
 \prod_{k=1}^{r} \alpha_k^{E(k,i)}  \in [-1,1]
$$
Therefore, we can replace this product with a new variable $\beta_i \in [-1,1]$. 
This results in an overapproximation because it drops dependencies that the factors $\alpha_k$ may have had with other terms.
This can be made slightly tighter if the exponent $E(k,i)$ is always even:
\begin{align*}
\prod_{k=1}^{r} \alpha_k^{E(k,i)} \in [0,1] \quad \quad \text{($E(k,i)$ is even for all $k$).}
\end{align*}
In this case, we replace the product with $\frac{ \beta_i + 1}{2}$, since a zonotope requires each $\beta_i \in [-1,1]$. Now we give the formal definition of the overapproximation step:

\begin{definition}\label{zonotope_approx}
Let $\cPZ = \cP_I \oplus \cPZ_D$ be a polynomial zonotope with $\cPZ_D = \langle G_D,E \rangle$.
The zonotope overapproximation of $\cPZ$ is defined as:
\begin{align}\label{Poly}
    \cP &= \cP_I \oplus \cP_D
\end{align}
with
\begin{align*}
\cP_D &= \Bigg  \{\sum_{i \in K} \beta_i G_D(\cdot ,i)  + \sum_{i \in H} \left ( \frac{\beta_i + 1}{2} \right )G_D(\cdot ,i)  ~~
     \bigg | ~~ \beta_i \in [-1,1] \Bigg \} \\
\end{align*}
where $H$ is the set of indices of terms with all even powers, 
$$H = \{ i ~  | ~\forall k ~ E(k,i) \equiv 0 \mbox{ \textnormal{(mod 2)} }  \}$$ 
and $K$ the set of remaining indices 
$$K =  \{1,\cdots, h  \}/H$$
\end{definition}

\vspace{1em}
\noindent
\textbf{Split Step.}
When the overapproximation of a polynomial zonotope is too large, step two of the algorithm splits the polynomial zonotope into two smaller pieces.  
This is done by choosing some factor $\alpha_s$ to split and then noting:
\begin{align}\label{split_intuition}
[-1,1] = \underbrace{\left\{\frac{1+\alpha_s}{2}~ \bigg | ~ \alpha_s\in [-1,1]\right\}}_{[0,1]} \cup \underbrace{\left\{-\frac{1+
\alpha_s}{2}~ \bigg |~  \alpha_s \in [-1,1]\right\}}_{[-1,0]}.    
\end{align}

\noindent
We split the dependent part of a polynomial zonotope $\cPZ_D = \langle G_D,E \rangle_{\cPZ}$ into:

\begin{align*}
    \cPZ_{D_{{1,s}}} &= 
     \Bigg \{ \sum_{i=1}^{h} 
    A_{s,i}
    \left (\frac{1+\alpha_s}{2} \right)^{E(s,i)}G_D(\cdot,i)~\Bigg | \quad \alpha_k \in [-1,1] \Bigg \}
\end{align*}
and 
\begin{align*}
    \cPZ_{D_{{2,s}}} &= \Bigg \{  \sum_{i=1}^{h} A_{s,i}
    \left (-\frac{1+\alpha_s}{2} \right )^{E(s,i)}G_D(\cdot,i) ~\Bigg | \quad \alpha_k\in [-1,1] \Bigg \}
\end{align*}
where $A_{s,i}$ is the product of all the factors in the $i$th term excluding $\alpha_s$:
\begin{equation}
A_{s,i} = \prod_{k=1,k\neq s}^{r}  \alpha_k^{E(k,i)}.
\label{eq:def:Asi}
\end{equation}
\begin{proposition}\label{Split}
    If $\cPZ = \cP_I \oplus \cPZ_D$, where $\cPZ_{1,s} =\cP_I \oplus  \cPZ_{D_{1,s}}$ and $\cPZ_{2,s} =\cP_I \oplus  \cPZ_{D_{2,s}}$, then $\cPZ_{1,s}$ and $\cPZ_{2,s}$ are polynomial zonotopes and 
    \[
    \cPZ = \cPZ_{1,s} \bigcup \cPZ_{2,s}.
    \]
\end{proposition}
\begin{proof}
This follows from the definitions of $\cPZ_{1,s}$ and $\cPZ_{2,s}$ using~\eqref{split_intuition}.
\end{proof}
\subsection{Convergence Concerns}
While the the overapproximate and split algorithm for polynomial zonotopes has been simply stated in prior work~\cite{kochdumper2022extensions,bak2022reachability}, we now identify two non-obvious concerns with the approach, given in Propositions~\ref{no-monotonic} and~\ref{unbounded_example}.

\vspace{1em}
\noindent
\textbf{Convergence Concern 1:} As the sizes of the polynomial zonotopes get smaller in the algorithm due to splitting, we may expect the corresponding zonotope overapproximation is also getting smaller. 
However, this is not true in general.
\begin{proposition}\label{no-monotonic}  
Overapproximation error can increase during the overapproximate and split algorithm.
\end{proposition}
\begin{proof}
Consider the following polynomial zonotope:
$$\cPZ = \big \{ \alpha_1^2 ~~ \big | ~~ \alpha_1 \in [-1,1]  \big \} = [0,1]$$
The overapproximation of $\cPZ$ is the zonotope
$$\cP = \bigg \{ \frac{1 + \beta_1}{2} ~~ \bigg | ~ \beta_1 \in [-1,1] \bigg \} = [0,1]$$ 
However, if we split $PZ$ into two parts using~\eqref{split_intuition}: 
$$\cPZ_1 = \cPZ_2 = \Bigg \{ \frac{1}{4} \bigg(1 + \alpha_1^2 + 2\alpha_1 \bigg) ~~ \Bigg |~~  \alpha_1 \in [-1,1] \Bigg \}$$
The overapproximation of $\cPZ_1$ and $\cPZ_2$ is
$$ \cP_1 = \cP_2 = \Bigg \{ \frac{1}{4} \bigg(\frac{3}{2}+ \frac{1}{2}\beta_1 + 2\beta_2 \bigg) ~~ \Bigg | ~~  \beta_1,\beta_2 \in [-1,1] \Bigg \} = \Big [-\frac{1}{4},1 \Big ]$$ 
Consequently, the original overapproximation $\cP \subset \cP_1 \bigcup \cP_2$.
After performing splitting, the union of the zonotope overapproximations became larger than the overapproximation before splitting---the overapproximation error has grown. \qed
\end{proof}
%
%

\vspace{1em}
\noindent
\textbf{Convergence Concern 2:} As shown above, overapproximation error can increase during the algorithm, although the individual split polynomial zonotopes are getting smaller.
Is there a bound between the the error of a polynomial zonotope and its overapproximation? 
No.

\begin{proposition}\label{unbounded_example}
The error between a polynomial zonotope and its zonotope over-approximation is unbounded. 
\end{proposition}
\begin{proof}
Consider the odd Chebyshev polynomials of first kind:
\begin{align*}
    &T_1(\alpha) = \alpha \\
    &T_3(\alpha) = 4\alpha^3 - 3\alpha \\
    &T_5(\alpha) = 16\alpha^5 - 20\alpha^3 + 5\alpha \\
    & \quad \quad \quad \ldots
\end{align*}
For Chebyshev polynomials, when $\alpha \in [-1,1]$ it is known that $T_{k}(\alpha) \in [-1, 1]$ (see Figure~\ref{fig:chebshev}).
However, the number of terms in the odd Chebyshev polynomials grows unbounded as $k$ increases. 
As a result, if we construct a polynomial zonotope from $T_{k}$ and overapproximate it with a zonotope using Definition~\ref{zonotope_approx}, the overapproximation also grows without bound as $k$ increases. \qed
\end{proof}

\begin{figure}[t]
    \centering    \includegraphics[width=\linewidth]{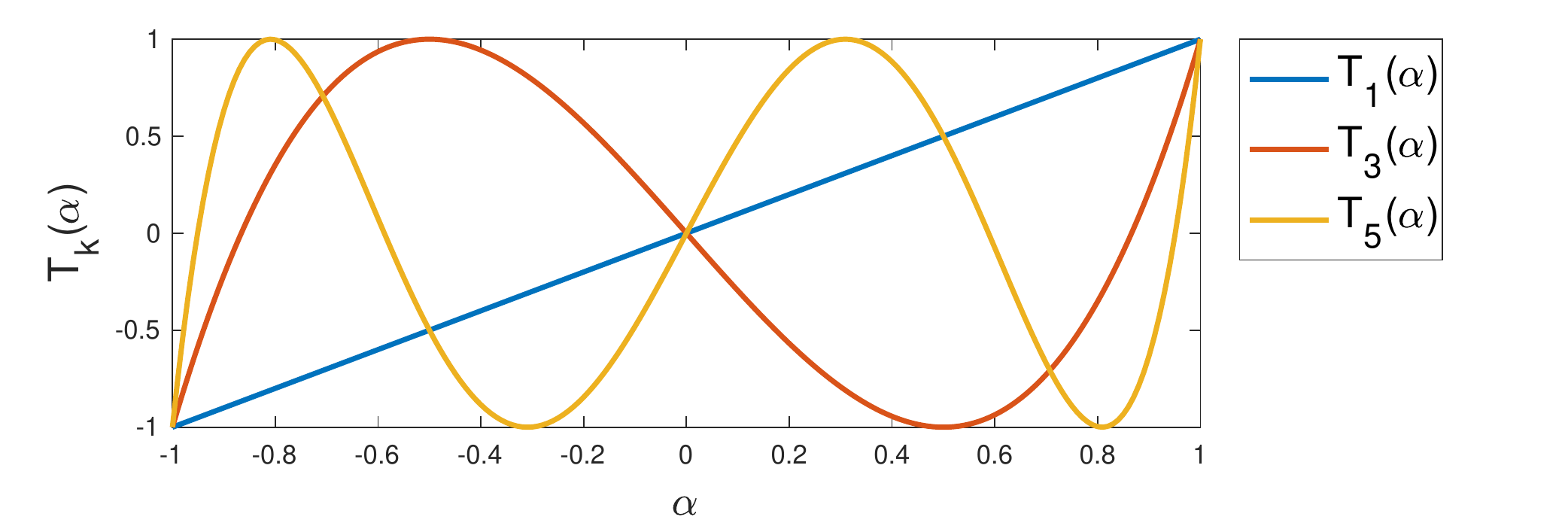}
    \caption{When $\alpha \in [-1, 1]$, the odd Chebyshev polynomials are between $[-1, 1]$, but the zonotope overapproximation using Definition~\ref{zonotope_approx} grows unbounded. }
    \label{fig:chebshev}
\end{figure}
Given Propositions~\ref{no-monotonic} and~\ref{unbounded_example}, there is a real concern that overapproximate and split algorithm may not always converge.
In the next section, we identify sufficient conditions where convergence can be guaranteed.

%
%
%
\section{Guaranteeing Convergence}
\label{sec:convergence}
While the overapproximate and split algorithm has been presented in prior work, as discussed in the previous section there is a real concern it may not always terminate.
In this section, we discuss conditions needed to ensure overapproximation error converges. 
First, we define the Hausdorff distance between two sets to serve as a criterion to evaluate the error of the zonotope overapproximation.
  \begin{definition}[Hausdorff distance]
   Given sets $S_1$ and $S_2$, the Hausdorff distance is:
   \begin{align*}
       d(S_1,S_2)= \max \Bigg \{ \sup_{x \in S_1} \inf_{y \in S_2} \Vert x - y \Vert,\sup_{y \in S_2} \inf_{x \in S_1} \Vert x - y \Vert \Bigg \}.
   \end{align*}
  \end{definition}
Note that in the case of nested sets $S_1\subset S_2$,  the second term is always 0, and the distance simplifies to
\[
d(S_1,S_2)= \sup_{x \in S_1} \inf_{y \in S_2} \Vert x - y \Vert.
\]
We now show that the Hausdorff distance between a polynomial zonotope and its zonotope overapproximation can be bounded using the norm of the generator matrix.
We will use the entry-wise matrix one-norm: $\Vert A \Vert = \sum_{j=1}^{m} \Vert A(\cdot,j) \Vert $.

  \begin{lemma}\label{relax}
Let $\cPZ = \cP_I \oplus \cPZ_D$ with dependent part $ \cPZ_D = \langle G_D,E \rangle _{PZ_D}$ and zonotope overapproximation $\cP = \cP_I \oplus \cP_D$ from Definition \ref{zonotope_approx},
\begin{align*}
    d(\cPZ,\cP) \leq \Vert G_D\Vert.
\end{align*}
\end{lemma}
\begin{proof}
Since $\cPZ\subset \cP$, then 
\begin{align*}
    d(\cPZ,\cP) 
    =\sup_{y \in \cP} \inf_{x \in \cPZ} \Vert x - y \Vert.
\end{align*}
Now for any point $y$ in the zonotope overapproximation we can write it as:
\begin{align*}
    y =  y_ I + y_D
\end{align*}
where $y_I \in Z_I$ and $$ y_D =   \left (\sum_{i \in K} \beta_i G_D(\cdot ,i)    +   \sum_{i \in H} \frac{\beta_i + 1}{2} G_D(\cdot ,i)\right )  \in Z_D.$$
Since $ y_I \in \cPZ$, we have:
\[
d(\cPZ,\cP) = \sup_{y \in \cP} \inf_{x \in \cPZ} \Vert x - y \Vert  
     \leq \sup_{y \in \cP} \Vert y_I - y \Vert  = \sup_{y \in \cP }\|y_D\|.
    \]
    Using the triangle inequality, 
    \[
    \Vert y_D \Vert = \bigg \Vert \sum_{i \in K} \beta_i G_D(\cdot ,i)    +   \sum_{i \in H} \frac{\beta_i + 1}{2} G_D(\cdot ,i)
    \bigg \Vert 
      \leq  \sum_{i \in H\bigcup K} \Vert G_D(\cdot,i) \Vert  = \Vert G_D \Vert
    \]
thus completing the proof. \qed
\end{proof}
Since the union of the split polynomial zonotopes forms the original polynomial zonotope (Proposition \ref{Split}), to demonstrate that the union of zonotope approximations converges to the original polynomial zonotope, we first establish that each zonotope overapproximation converges to its respective polynomial zonotope. 
Using Lemma \ref{relax}, we only need to show that the norm of the dependent matrix decreases after splitting. 

In order to show this, we need additional constraints on how this splitting variable $s$ is chosen.
Various heuristics for choosing $s$ can be found in the literature, but to ensure convergence we require a \emph{fairness} assumption, in that each factor needs to be selected an infinite number of times. 
We assume $s$ is chosen cyclically to satisfy this requirement.
Let us first consider a simple example which tells us why the norm of the dpendent matrix decreases after splitting cyclically.
\begin{example}
Let us first consider a polynomial zonotope 
$\cPZ = \{ \alpha_1^E ~ \big | ~ \alpha_1 \in [-1,1] \} $ which only has a single factor $\alpha_1$, the dependent matrix $G_D = 1$ and exponent $E$ In this case, when split, 
\begin{equation}
\left(\pm \frac{1+\alpha}{2} \right)^{E} = \frac{\pm 1}{2^E} + \frac{1}{2^{E}}\sum_{j=1}^{E} \binom{E}{j} (\pm 1)^j\alpha_k^j
\label{eq:binomsplit:helper}
\end{equation}
and since the constant will belong to independent part, dependent generator becomes 
\[
G_D^{1,1} = \frac{1}{2^E} \begin{bmatrix}
    \binom{E}{1}, & \binom{E}{2}, & \cdots, & \binom{E}{E}
\end{bmatrix},\qquad
G_D^{2,1} = \frac{1}{2^E} \begin{bmatrix}
    -\binom{E}{1}, & \binom{E}{2}, & \cdots, & (-1)^E\binom{E}{E}
\end{bmatrix}
\]
and thus the generator norm has shrunk 
\[
\|G_D^{j,1}\| = 1 - \frac{1}{2^E}.
\]    
\end{example}
Next we present the more general result.
\begin{lemma}\label{onestepite}
    Let $\cPZ$  be a given polynomial zonotope with dependent part $\cPZ_D = \langle G_D,E \rangle$, with $r$ factors and $h$ generators.
    Assuming cyclical splitting, after splitting $s$ times, we have $2^s$ polynomial zonotopes $\cPZ_{1}^{s},  \cPZ_{2}^{s}, \cdots, \cPZ_{2^s}^{s}$. 
    When $s < r$, the norm cannot increase
    \begin{align*}    
     \Vert G_D^{j,s} \Vert \leq  \Vert G_D\Vert
    \end{align*}
    When $s = r$, the norm decreases by a factor $\rho < 1$,
    \begin{align*}
        \Vert G_D^{j,s} \Vert \leq \rho \Vert G_D \Vert
    \end{align*}
    where $\rho = \max_{i \in \{1,2\cdots,h\}} \bigg (1 - (\frac{1}{2})^{\Vert E(\cdot,i)\Vert } \bigg )$ , $j \in \{1,2,\cdots,2^s \}$ and $G_D^{j,s}$ is the dependent factor generator of $\cPZ_{j}^{s}$
\end{lemma}
Importantly, the factor $\rho$ does not depend on the number of splits, ; it only depends on the original $\cPZ$.
\begin{proof} 
    When $ s < r$, let us consider the dependent part of $\cPZ_{j}^{s}$, it will have the following form:
    \begin{align} \label{eq:depless}
    \left (\cPZ_{j}^{s} \right )_D&=  \Bigg \{ \sum_{i=1}^{h} 
    A_{s,i}^{j} \zeta_{j}^{i} \prod_{k=1}^{s} \left (\frac{1+\alpha_k}{2} \right)^{E(k,i)}G_D(\cdot,i)~\Bigg | \quad \alpha_k \in [-1,1] \Bigg \}  
    \end{align}
    where 
    $\zeta_j^i \in \{-1,1\}$ distinguishes between the $2^s$ polynomial zonotopes based on which side of each factor was chosen while splitting,
    and $A_{s, i}^{j}$ is the product of factors that have not yet been split: $$A_{s, i}^{j} = \prod_{k=s+1}^{r} \alpha_k^{E(k,i)}$$ 
    Now because $\vert a_k  \vert \leq 1$, we have both $\vert A_{s,i}^{j} \vert \leq 1$ and
    $$
     \quad \Bigg \vert \prod_{k=1}^{s} \left (\frac{1+\alpha_k}{2} \right)^{E(k,i)} \Bigg  \vert \leq 1
    $$
    As a result the absolute value of their product, the value that multiplies $G_D(\cdot, i)$ in Equation~\ref{eq:depless} is also less than $1$. 
    Therefore, when $s < r$ for any $j \in \{1,2,\cdots, 2^s \}$, 
    we have:
    \begin{align*}    
        \Vert G_D^{j,s} \Vert 
        &\leq \sum_{i=1}^{h} \Vert G_D(\cdot,i) \Vert  =  \Vert G_D \Vert 
    \end{align*}
    
    \noindent
    Next, in the other case when  $ s = r$ we can expand the exponent:
    \begin{align*}
        \prod_{k=1}^{r} \left (\frac{1+\alpha_k}{2} \right)^{E(k,i)} &= \sum_{\xi_1,\xi_2,\cdots,\xi_r} c_{\xi_1,\cdots,\xi_k} \prod_{k=1}^{r}\alpha_k^{\xi_r}
    \end{align*}
    where all the coefficients are positive and the first one $c^i_{0,\cdots,0} = \left (\frac{1}{2} \right )^{^{\Vert E(\cdot,i)\Vert}}$. 
    By taking $\alpha_k = 1$ for all $\alpha_k$, we obtain:
    $$
    \sum_{\xi_1,\cdots,\xi_k} c^i_{\xi_1,\cdots,\xi_r} = 1
    $$
    As a result, we have 
    \begin{align*}
    &  \sum_{i=1}^{h} 
    A_{s,i}^{j}  \prod_{k=1}^{s} \left (\frac{1+\alpha_k}{2} \right)^{E(k,i)}G_D(\cdot,i) 
    = \sum_{i=1}^{h} 
    \zeta_j^i \prod_{k=1}^{r} \left (\frac{1+\alpha_k}{2} \right)^{E(k,i)}G_D(\cdot,i) \\
    &= \underbrace{\left (\sum_{i=1}^{h} \zeta_j^i  \Big( c^i_{0,\cdots,0} G_D(\cdot,i) \right )}_{\mbox{ (constant)}} + \left (\sum_{i=1}^{h} \zeta_j^i\sum_{\xi_1,\cdots,\xi_r}^{ \left (\sum_{k=1}^{r}\xi_k \right)\geq 1 }  c_{\xi_1,\cdots,\xi_r}\prod_{k=1}^{r} \left( \alpha_k \right)^{\xi_k} G_D(\cdot,i) \right ) 
    \end{align*}
    In this case, the constant part will not be in the dependent part of $\cPZ_j^r$ but will be moved to the independent part, so that:
    \begin{align*}
    \left (\cPZ_{j}^{r} \right )_D &=  \Bigg \{ \sum_{i=1}^{h} \zeta_j^i\sum_{\xi_1,\cdots,\xi_r}^{ \left (\sum_{k=1}^{r}\xi_k \right)\geq 1 }  c^i_{\xi_1,\cdots,\xi_r}\prod_{k=1}^{r} \left( \alpha_k \right)^{\xi_k} G_D(\cdot,i) \Bigg | ~~ \alpha_k \in [-1,1]
    \Bigg \} 
    \end{align*}
    For any $j \in \{1,2,\cdots, 2^s \}$, we will have:
      \begin{align*}
        \Vert G_D^{j,r} \Vert 
            &= \Bigg \Vert \sum_{i=1}^{h} \zeta_j^i\sum_{\xi_1,\cdots,\xi_r}^{ \left (\sum_{k=1}^{r}\xi_k \right)\geq 1 }  c^i_{\xi_1,\cdots,\xi_r}G_D(\cdot,i) \Bigg \Vert  \\
            &\leq   \sum_{i=1}^{h} \underbrace {\sum_{\xi_1,\cdots,\xi_r}^{ \left (\sum_{k=1}^{r}\xi_k \right)\geq 1 }  c^i_{\xi_1,\cdots,\xi_r}}_{ 1 -  (\frac{1}{2})^{\Vert E(\cdot,i)\Vert}}\Vert G_D(\cdot,i)  \Vert
            \\
        &= \sum_{i=1}^{h}  \Bigg (1 - \bigg (\frac{1}{2} \bigg)^{\Vert E(\cdot,i)\Vert} \Bigg ) \Vert G_D(\cdot,i) \Vert  \\
        & \leq \rho \sum_{i=1}^{h} \Vert G_D(\cdot,i) \Vert = \rho \Vert G_D \Vert
    \end{align*} \qed
\end{proof}
\begin{corollary}\label{lstepite}
 Let $\cPZ$  be a given polynomial zonotope with dependent part $\cPZ_D = \langle G_D, E \rangle$. 
 Using cyclical splitting, after splitting $s$ times, the factor with index $ 1 + \big ((s- 1) (\mbox{mod } r) \big )$ will be split. 
 Let $\cPZ_{1}^{s} ,  \cPZ_{2}^{s} ,\cdots, \cPZ_{2^s}^{s}  $ be the split polynomial zonotope after $s$ iterations. Then for any $ 0 < j \leq 2^s$
    \begin{align*}    
    \Vert G_D^{j,s} \Vert \leq  \rho^{\lfloor s/r \rfloor}  \Vert G_D \Vert
    \end{align*}   
    where  $G_D^{j,s}$ is the dependent factor generator matrix of $\cPZ_{j}^{s}$.
\end{corollary}
With this, we can now show that the union of the zonotope overapproximations converges to the original polynomial zonotope.
  %
  %
  \begin{theorem}\label{convergence_cyc}
         Let $\cPZ$  be a given polynomial zonotope with dependent part $\cPZ_D = \langle G_D,E \rangle$. 
         Using cyclical splitting, after splitting $s$ times, the factor with index $ 1 + \big ((s- 1) (\mbox{mod } r) \big )$ will be split. 
         %
         The corresponding split polynomial zonotopes are $\cPZ_1^{s},\cPZ_2^{s},\cdots,\cPZ_{2^{s}}^{s}$ and zonotope overapproximation for $\cPZ_j^{s}$ is $\cP_j^{s}$.
         As we split more often, the overapproximation error converges to zero:
    \begin{align*}
       \lim_{s \to \infty} d \bigg(\cPZ,\bigcup_{j = 1}^{2^{s}}\cP_j^{s} \bigg) = 0
   \end{align*}  
  \end{theorem}

\begin{proof}
        Since $\cPZ \subseteq \bigcup_{j = 1}^{2^{s}}\cP_j^{s}$,
      \begin{align*}
      d(\cPZ,\bigcup_{j = 1}^{2^{s}}\cP_j^{s}) 
      &= \sup_{y \in \bigcup_{j = 1}^{2^{s}}\cP_j^{s}} \inf_{x \in \cPZ} \Vert x - y \Vert \\
      & = \max_{j} \sup_{y \in \cP_j^{s}} \inf_{x \in 
      \cPZ} \Vert x - y \Vert \\
      & \leq \max_{j} d(\cP_j^{s},\cPZ_{j}^{s}) \\
      & \leq \rho^{\lfloor s / r \rfloor} \Vert G_D \Vert \quad \mbox{(by Lemma \ref{onestepite} and Corollary \ref{relax}).}\\ 
      \end{align*}
      Since $0 < \rho < 1$, the value of limit of $\rho^{\lfloor s / r \rfloor}$ converges to zero. \qed
\end{proof}

Although we have been assuming cyclical variables splitting, the above theorem could be adapted to more general splitting schemes by noting that the generator matrix norm must reduce after every factor has been selected at least once.
As long as the spitting approach is fair, in the sense that it does not ignore any factors forever, the dependent generator matrix norm will  decrease by a factor of $\rho$ after each full round.
Applied repeatedly, the norm of the dependent matrix will therefore  decrease log-linearly to 0, in terms of the number of rounds. 


\vspace{1em}
Lastly, in Figure~\ref{fig:split_comparison}, we motivated our work with a practical example where the overapproximate and split algorithm did not appear to converge to the true polynomial zonotope.
Based on our results in Theorem~\ref{convergence_cyc}, we know that convergence is guaranteed, but you may need to split along each dependent factor.
In the polynomial zonotope in the figure, the number of dependent factors was around 50, so even after 40 splits the overapproximation error can remain large.

\section{Related Work}
Polynomial zonotopes were originally designed to represent non-convex sets to tightly enclose the reachable sets for a nonlinear system~\cite{althoff2013reachability}. 
A sparse version of the representation was proposed in follow up work \cite{kochdumper2020sparse,kochdumper2022extensions} to support a more compact representation while still being closed under key nonlinear operations.
Recent extensions add linear constraints to the domain which are called constrained polynomial zonotopes~\cite{kochdumper2020constrained}.
Although this was not the focus of the current paper, the intersection and plotting algorithms for constrained polynomial zonotopes is basically the same as for polynomial zonotopes, except the overapproximation step results in constrained zonotopes~\cite{scott2016constrained} (also called star sets~\cite{duggirala2016parsimonious} or $\mathcal{AH}$-Polytopes~\cite{sadraddini2019linear}) rather than zonotopes. 
Therefore, we expect the analysis results from this paper to also be transferable to constrained zonotopes.

Besides reachability analysis of nonlinear systems, polynomial zonotopes have been used for reachability of linear systems with uncertain parameters~\cite{Luo2023} which resulted in more accurate reachable sets comparing to zonotope methods.
The representation has also been used for set-based propagation through neural networks~\cite{kochdumper2023nfm},
real-time planning and control scenarios~\cite{michaux2023can}
and safety shielding for reinforcement learning systems~\cite{kochdumper2023provably}.

In cases where the model of the dynamical system is not given, polynomial zonotopes can also be used for reachability with Koopman linearized surrogate models obtained from trajectory data \cite{bak2022reachability}. 
Since Koopman linearization requires lifting the state through a nonlinear transformation, convex initial sets in the original space can become complex non-convex sets.
Polynomial zonotopes can provide tight enclosures of these lifted initial sets. 
%

Taylor models \cite{makino2003taylor} are a related set representation sometimes used for reachability analysis~\cite{chen2013flow} that are similar to polynomial zonotopes with interval remainders added to each variable.
Taylor model arithmetic allows one to approximate arbitrary smooth functions, although the intersection and plotting algorithms are essentially grid pavings over the domain of the set.

As mentioned in the introduction, polynomial zonotope intersection checking is equivalent to the box-constrained polynomial optimization problem.
There are several methods to solve such problems, for example augmented Lagrangian methods or sum of squares programming\cite{shor1987class,shor1987class,parrilo2003semidefinite,parrilo2000structured,lasserre2001global}.
Augmented Lagrangian methods consider box constrained polynomial constraint problems as a general nonlinear programming problem. 
General nonlinear programming with convex constraints can usually be solved by considering the KKT conditions \cite{bazaraa2013nonlinear,boyd2004convex,rockafellar1997convex}. 
The KKT points can be found by augmented Lagrangian methods\cite{hestenes1969multiplier,powell1969method}.
The exact augmented Lagrangian methods must solve a subproblem in each update. 
Hence the inexact augmented Lagrangian methods(iALM) are used in practice. 
Although there are many works on iALM~\cite{li2021rate,li2021augmented}, such methods guarantee local convergence with local optimal solutions. Therefore, they are not commonly used for polynomial optimization problems.
Sum-of-squares polynomials are polynomials that can be formulated as the sum of the square of several polynomials. 
If the polynomial can be formulated in this way, or reformulated after a series of liftings \cite{lasserre2007sum}, then the optimization problem can be formulated as semidefinite programming and solved using convex optimization (although the resulting problem may be very large). 
%

\section{Conclusions} 

In this work we discussed the difficulty of the fundamental intersection checking operation for the polynomial zonotope set representation.
This difficulty is rarely directly addressed in papers that use polynomial zonotopes, although it can be a practical limitation of any algorithm that builds upon the set representation.
The complexity is both theoretical and practically relevant, as we have shown cases, specifically Figure~\ref{fig:split_comparison}, where accurate approximation using the overapproximate and split approach is intractable.
While polynomial zonotopes are a powerful tool for formal verification, they are not a panacea, as much of the problem complexity can be often hidden within the representation itself, manifesting when performing set intersections.

\bibliographystyle{splncs04}
\bibliography{refs}
\end{document}